\documentclass[12pt]{extarticle}
\usepackage[utf8]{inputenc}
\usepackage{amsmath}
\usepackage[margin=01.0in]{geometry}
\usepackage{amssymb}
\usepackage{amsfonts}
\usepackage{amsthm}
\usepackage{enumitem}

\usepackage{pgfplots}
\pgfplotsset{compat=1.8}
\usepackage{tikz}
\usepackage{amsmath}
\usepackage{xcolor}
\usepackage{float}

\usetikzlibrary{decorations.markings}
\usetikzlibrary{snakes}

\tikzset{
    thickest/.style={line width=3pt},
    empty/.style={decoration={markings,
    mark=at position #1 with {\fill[white,draw=cyan,thick] circle (3pt);}},postaction={decorate}},
    full/.style={decoration={markings,
    mark=at position #1 with {\fill circle (3pt);}},postaction={decorate}},
}

\numberwithin{equation}{section}

\usepackage{tabularx}

\usepackage{graphicx}
\graphicspath{ {./images/} }

\theoremstyle{plain}
\newtheorem{thm}{Theorem}
\newtheorem{lemma}[thm]{Lemma}
\newtheorem{cor}[thm]{Corollary}

\newtheorem{prop}[thm]{Proposition}

\theoremstyle{definition}

\theoremstyle{remark}
\newtheorem{rem}[thm]{Remark}

\numberwithin{thm}{section}

\definecolor{applegreen}{rgb}{0.55, 0.71, 0}
\definecolor{aqua}{rgb}{0, 1.0, 1.0}

\title{Rationalizability, Iterated Dominance, and
 the Theorems of Radon and Carathéodory}
\author{Roy Long \thanks{University of Chicago Economics \\ Originally written February 2022. This version January 2024}}
\date{February 2022 \\ \bigskip \bigskip \normalsize{\emph{Stanford Economic Review 2024}}  } 

\begin{document}

\maketitle

\begin{abstract}
The game theoretic concepts of rationalizability and iterated dominance are closely related and provide characterizations of each other. Indeed, the equivalence between them implies that in a two player finite game, the remaining set of actions available to players after iterated elimination of strictly dominated strategies coincides with the rationalizable actions. I prove a dimensionality result following from these ideas.
I show that for two player games, the number of actions available to the opposing player provides a (tight) upper bound on how a player's pure strategies may be strictly dominated by mixed strategies. I provide two different frameworks and interpretations of dominance to prove this result, and in doing so relate it to Radon's Theorem and Carathéodory's Theorem from convex geometry. These approaches may be seen as following from point—line duality. A new proof of the classical equivalence between these solution concepts is also given.

\end{abstract}
\newpage

\section{Introduction}

Strategic dominance is a standard concept in game theory. It encapsulates the intuition that in a strategic environment, one should never take some action $a$ if in all possible scenarios (formally, all profiles of the other players' actions), another action $a'$ does strictly better.

In this section, we review the concepts of dominance and rationalizability, and summarize the content of the paper. \\

We consider finite, two player strategic games. For $i=1,2$, player $i$'s (pure) strategies is the set of actions $A_i$.

\textbf{Definition:} A pure strategy for player $i$, $a_i \in A_ii$ is strictly dominated by another pure strategy $a_i'$ if $u_i(a_i, a_{-i}) < u_i(a_i', a_{-i})$ for any $a_{-i}.$

Dominance can be generalized to mixed strategies, where players randomize over actions. In this case, the expected payoff is compared.

\textbf{Definition:} A (pure) strategy $a_i \in A_i$ is strictly dominated by a mixed strategy $s_i$ over the set $\{a_j\}$ if $u_i(a_i, a_{-i}) < U_i(s, a_{-i}) = \sum_{j} p_j u_i(a_j, a_{-i})$ for all $a_{-i}$ where $(p_j)$ is a probability distribution over $\{a_j\}$ specified by $s.$

If the game is finite, the iterated elimination of strictly dominated strategies may be performed. The simple algorithm is to repeatedly remove strategies (for either player) which are strictly dominated by some mixed strategy from the game, until no such strategies are left. It is easy to show that the set of remaining strategies are those that were not strictly dominated in the original game, and any strictly dominated strategy is removed by this process.

More recently, the solution concept of \textit{rationalizability} was developed. For two player games the definition is as follows:

\textbf{Definition:} An action $a_i \in A_i$ is \textit{not rationalizabile} if for any belief $q = (q_1, \ldots, q_m)$ player $i$ has about player $(-i)$'s strategy $s_{-i},$ there exists some other action $a_i' \in A$ (which may depend on $q$) such that $U_i(a_i, s_{-i}) < U_i (a_i', s_{-i})$

One can also give a recursive definition of $k$-level rationalizability ($k \in \mathbb{N}$), and repeatedly remove strategies that are not rationalizable until all remaining actions are rationalizable.

It is easy to show that any strictly dominated strategy is not rationalizable, as one intuitively expects. It turns out that the converse is true as well, which was first shown in 1984 independently by both Bernheim [1] and Pearce [2].

\begin{thm} (Pearce, Bernheim).
    Given a finite, two player strategic game $G,$ the set of (pure) strategies that remain after iterated elimination of strictly dominated strategies are precisely the strategies that are rationalizable in $G.$
\end{thm}

\begin{proof}
    See Appendix.
\end{proof}

Bernheim's proof represents strategies as payoff vectors in $\mathbb{R}^d$ and uses the separating hyperplane theorem. Pearce's proof uses the minimax theorem for zero sum games. See Fudenberg [3] or Yildiz [4] for a presentation of the former and Obara [5] for the latter.

The central result of this paper is $\textbf{Theorem 2.1}.$ We prove a tight bound on the minimum number of strategies needed to form a mixed strategy that strictly dominates another strategy. There is a rich connection to convex geometry in our analysis, and there are several representations of dominance and rationalizability that lead to classical results in this area of mathematics.

The methods in section 2 also yield a new proof of $\textbf{Theorem 1.1}$, mirroring the separating hyperplane proof of the equivalence. We defer our proof to the appendix.

In section 2, we provide a proof of $\textbf{Theorem 2.1}$ using Radon's Theorem motivated by a natural view of the equivalence between rationalizability and dominance. We also show that the bound still holds when a player can have infinitely many actions.

In section 3, we represent dominance using vectors in $\mathbb{R}^d.$ We give another proof of $\textbf{Theorem 2.1},$ and discuss the mathematical connections between the two proofs.

 We may also the two approaches as resulting from a kind of point—line duality. In section 3, we view strategies as payoff vectors (or points), while in section 2, strategies are represented by hyperplanes. 

\section{Connections to Convex Geometry}

We study the following question: Given a two player finite strategic game $G$ with action sets $A, B,$ suppose some player's (without loss of generality player 1) pure strategy $a_i \in A$ is strictly dominated by a mixed strategy on a set of actions $\{a_j\} = A' \subseteq A.$ Can we bound the size of the support of this mixed strategy, $|A'|$?

An obvious bound is $|A'| \leq |A|,$ and furthermore $|A'| \leq |A| - 1,$ as if $a_i \in A'$ then $a_i$ is still dominated by some mix on $A' \setminus \{a_i\}.$

Interestingly, it turns out that we actually can bound $|A'|$ by the number of actions available to player 2, i.e. $|B|.$
This is not at all obvious.
\\

More precisely, we have the following result:

\begin{thm}
Let $G$ be a finite two player game, where player $1$'s set of actions is $A$ and player $2$'s set of actions is $B.$ If $a_i \in A$ is strictly dominated by a mix of strategies over $\{a_j\} \subseteq A,$ then in fact $a_i$ is strictly dominated by a mixed strategy $\{a_j'\}$ consisting of at most $\min(|A|,|B|)$ actions. An analogous bound holds for player $2.$
\end{thm} 

\begin{rem}
    The bound for player 1 can be improved to $\min(|A|-1, |B|)$ based on the simple observation above. In fact, we will show that this bound is tight.
\end{rem}

Now the result for player 1 is  trivial if $|A| \leq |B|,$ but it is not immediately clear why such a result should be true otherwise. Why does the number of actions the opponent has affect (and moreover provides a bound on) how a player's pure strategies are dominated by mixes of their own strategies?
\\
\newline
For example, if player $2$ has $|B| = 4$ available actions and player $1$ has $|A| = 10$ available actions, then this tell us that if some action $a_i \in A$ of player $1$ is strictly dominated, it's in fact dominated by some set of just $4$ of player $1$'s available actions.
\\
\newline
When $|B| = m = 1$, it is trivial, as player $2$ only has one action that she must play. When $m=2,$ this tells us that if a pure strategy is strictly dominated, either it is dominated by another pure strategy or by a mix of two strategies.
\\
\newline
Henceforth, we assume von Neumann Morgenstern preferences, i.e. preferences over strategies are compared based on their expected payoffs.
\\
\newline
To motivate this theorem, we start by first examining the case for $m=2.$ Consider the simple $3 \times 2$ game, with the payoffs as listed, where rows are player $1's$ actions and columns are player $2's$ actions.
\\

\begin{figure}[H]
    \centering
\begin{tabularx}{0.3\textwidth}
{
  | >{\centering\arraybackslash}X 
  | >{\centering\arraybackslash}X
  | >{\centering\arraybackslash}X| }
  \hline
  1 $\backslash$ 2 & L & R \\
 \hline
 U & (6,1) & (0,3) \\
 \hline
 M & (2,1)  & (5,0)  \\
\hline
D & (3,2) & (3,1) \\
\hline 
\end{tabularx} \par 
\bigskip
\caption{payoff matrix}
\end{figure}

Then if player $1$ has belief $(q,1-q)$ over $\{L, R\}$ about player $2$'s strategy, then player $1$'s expected payoffs by playing $U,M,D$ are, respectively: \\
\newline
$E_U = 6q$
\newline
$E_M = 2q + 5(1-q) = 5-3q$
\newline
$E_D = 3$ \\
\\
By plotting these as graphs on the domain $q \in [0,1],$ we can see that for every $q \in [0,1],$ the $E_D(q)$ is less than at least one of $E_U(q)$ and $E_M(q).$ rationalizability then tells us that $D$ is strictly dominated by some mixed strategy of $U$ and $M.$
\\
\newline
Let $E_U = f_1, E_M = f_2, E_D = f_3$ for ease of notation. Since these functions are all linear and therefore convex, their maximum is also convex. So by taking the upper convex hull of the lines, ie. $f_{MAX} = \max_{1\leq i \leq3} \{f_i\},$ the graph of $f_3$ lies entirely below the upper convex hull $f_{MAX}$ and is strictly dominated by it.
\\
Now let's look at a different game, where player $2$ again has two actions $B = \{b_1, b_2\}$, and suppose player $1$ has five actions $\{a_1,\ldots,a_5\}$ corresponding to the following expected payoffs for player 1, if she has belief $(q,1-q)$ over $\{b_1, b_2\}$: \\
\newline
$E_1(q) = 1.2q+0.4$
\newline
$E_2(q) = -1.3q+1.3$
\newline
$E_3(q) = 0.5q+0.8$
\newline
$E_4(q) = -0.8q+1$
\newline
$E_5(q) = 0.8.$
\newline

The graph below shows the expected payoffs of these strategies as a function of the belief $q \in [0,1],$ where the black curve is $E_2,$ the green curve is $E_4,$ the horizontal purple curve is $E_5,$ the positive sloping purple curve is $E_1,$ and the blue curve is $E_3.$ \\

\begin{figure}[H]
    \centering
    \includegraphics[scale=0.7]{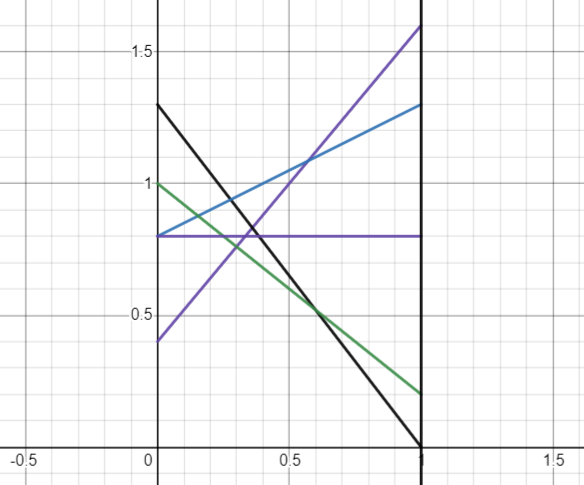}
    \caption{expected payoffs given beliefs}
\end{figure}

We can see that $a_4$ and $a_5$ are strictly dominated by $a_1,a_2,a_3$ as \\ $E_4(q), E_5(q) < \max\{E_1(q), E_2(q), E_3(q)\}$ for all $0 \leq q \leq 1.$
For example, $a_4$ is strictly dominated by a mix over $a_1,a_2,a_3$ with mixed strategy $q = (0.2,0.3,0.5),$ which results in an expected payoff of $0.2(1.2q+0.4)+0.3(-1.3q+1.3)+0.5(0.5q+0.8) = 0.1x+0.87 > 0.8 = E_5(q)$ for all $q \in [0.1]$, and $a_5$ is strictly dominated by a mix over $a_1,a_2,a_3$ with mixed strategy $q = (0.2,0.7,0.1)$ which results in an expected payoff of $0.2(1.2q+0.4)+0.7(-1.3q+1.3)+0.1(0.5q+0.8) = 1.07 - 0.62q > -0.8q + 1$ for all $q \in [0,1].$
\\

Notice that in both cases, we only need to choose two out of the three  to guarantee a mixed strategy that strictly dominates them. For example, if we consider a mix over with $a_1, a_2$ with probability distribution $(0.3,0.7),$ then the expected payoff is $0.3(1.2q+0.4) + 0.7(-1.3q+1.3) = 1.03 - 0.55 q > -0.8q + 1 = E_4(q)$ for all $q \in [0,1].$ Similarly, $a_5$ can be dominated by mixes over just $\{a_1, a_2\},$ as well as a mix of $\{a_2, a_3\}.$ Graphically, we see that for a strictly dominated action, we can choose just two of the graphs of other functions such that the maximum of them lies strictly above the dominated strategies' graph.
\\
\newline
Unfortunately, there are some edge cases where strict dominance never happens and the graph of every $f_i$ is part of the upper convex hull $\max f_i$ at some point.
\\
\newline
Thus, the equivalence of rationalizability and elimination of dominated strategies gives us the following result:

\begin{prop}
\textbf{(Theorem 2.1 for $d=2$)} 
Let $f_1,f_2,\ldots,f_k : [0,1] \to \mathbb{R}$ be linear, i.e. $f_i(x) = a_i x + b_i$ for all $i,$ and $q \in [0,1].$ Suppose $g : [0,1] \to \mathbb{R}$ is also linear and satisfies: for all $x \in [0,1],$ the $\max_{i} (f_i(x)) > g(x).$ Then there exists weights $0 \leq r_1,\ldots,r_k \leq 1$ such that $r_1+\cdots +r_k=1$ and $\sum r_i f_i(x) > g(x)$ for all $x \in [0,1].$ In fact, it is possible to choose the $r_i$ such that all but at most two of them are $0.$
\end{prop}

In other words, if the graph of an action's expected payoff lies strictly below the upper convex hull (i.e. it's strictly dominated), it is in fact strictly dominated by a mix over a set of at most two actions.

\begin{proof}
The first part of the statement follows from the equivalence of rationalizability and iterated dominance, so we show the second part.
First, note that we may assume that $g(x) = 0$ for $x \in [0,1],$ as we consider $f_i'(x) = f_i(x) - g(x)$ for all $i,$ which are still linear functions and the condition becomes $\max_{i} f_i' > 0$ with $g' = 0.$

First, remove all functions $f_i$ where $f_i(x) \leq g(x)$ for all $x \in [0,1]$, as they will never be part of the upper convex hull.
Now, for the graphs of the remaining functions, either they intersect the graph of $g$ or they don't.
If they don't, then they lie entirely above $g,$ i.e. $f_k(x) > g(x)$ for all $x \in [0,1]$, and we can just let $r_k = 1$ and all other $r_i = 0.$ \\
\newline
Otherwise, they have a unique intersection point.
For simplicity, as we noted above, we may assume $g$ is zero.
Then, we can look at where each $f_k$ intersects $[0,1]$ on the $x$-axis.
Let $L$ be the leftmost $x$-intercept in $[0,1]$ such that the corresponding $f$ has positive slope, and $R$ be the rightmost $x$-intercept in $[0,1]$ such that the corresponding $f$ has negative slope. Then, it's easy to see that $L < R$ because otherwise $\max f$ is not greater than zero at points between $L$ and $R$,
so it reduces to only having two functions $f_1,f_2,$
where $f_1 = a_1(x + b_1)$ and $f_2 = a_2(x+b_2)$ where $a_1 > 0, a_2 < 0$ and $0 \leq b_1 < b_2 \leq 1.$ \\
\newline
Then, we have $r_1 f_1 + r_2 f_2 = r_1 a_1(x+b_1) + r_2 a_2(x+b_2)$, and by setting $r_1 a_1 = -r_2 a_2$, we just get a constant function which is greater than zero ($r_1 f_1 + r_2 f_2$ goes through the intersection point of $f_1, f_2$ which lies above the $x$-axis.)
\end{proof}

Notice that in the last step of the previous proof, we didn't need to explicitly construct the function $r_1 f_1 + r_2 f_2.$ In fact, the equivalence of rationalizability and iterated strict dominance tells us if the graph of $g(x)$ is strictly below the upper convex hull formed by some set of $f_i,$ then $g$ is strictly dominated by those $f_i,$ so in particular, we're guaranteed that such $r_i$ exists for the corresponding $f_i$ that make the weighted sum greater than $g.$ The main result of the previous proof was that we only needed to choose two of the $f_i$ and it would be sufficient that the $\max f_i$ over those two functions would be greater than zero at all points.
\\
\newline
Now, we generalize the preceding intuition to higher dimensions, i.e. when player $2$ has more than two actions.
In essence, what we did in the previous proof was this: after subtracting $g$ from all the $f_i,$ we were left with graphs of linear functions whose overall maximum is strictly greater than $0,$ i.e. their upper convex hull lies strictly the $x$-axis. Then, we looked at all linear functions that intersected $[0,1].$ We can assume this because if the linear function had an $x$-intercept outside of $[0,1],$ then its restriction to $[0,1]$ will clearly be positive. Thus, we only needed to consider the cases of positive and negative slope that intersected $[0,1].$
Then, if a positive sloped line intersected $[0,1]$ at $x_0,$ the $\max_{i} f_i(x) > 0$ for all $x > x_0,$ and if a negative sloped line intersected $[0,1]$ at $x_0,$ we know $\max_{i} f_i(x) > 0$ for all $x_0 < 0.$
\\

Thus, the problem is reduced to: given a set of $1$-dimensional rays or half lines that cover $[0,1],$ we want to show that there exists a subset of at most two of them that also covers $[0,1].$ \\

\begin{figure}[H]
    \centering
    \begin{tikzpicture}[>=stealth, scale=1.5]
    \draw[<->, very thick] (-2,0)--(3,0);
    \draw (-0.3,0.1) -- (-0.3,-0.1);
    \draw (1.3, 0.1) -- (1.3, -0.1);

    \node[below] at (-0.3, -0.2) {$0$};
    \node[below] at (1.3, -0.2) {$1$};


    \draw[empty=0,-stealth, line width=1pt, color=cyan] (0.2,0.25) -- (3,0.25);

    \draw[empty=0,-stealth, line width=1pt, color=cyan]
    (1.0, 0.5) -- (3, 0.5);

    \draw[empty=0,-stealth, line width=1pt, color=cyan] 
    (0.7, 0.75) -- (3, 0.75);

    \draw[empty=0,-stealth, line width=1pt, color=cyan]
    (-0.1, 0.75) -- (-2, 0.75);

    \draw[empty=0,-stealth, line width=1pt, color=cyan]
    (0.4, 0.5) -- (-2, 0.5);

    \draw[empty=0,-stealth, line width=1pt, color=cyan]
    (1.5, 1.0) -- (3, 1.0);

    \end{tikzpicture}
    \caption{Covering of $[0,1]$ by half-lines}
\end{figure}
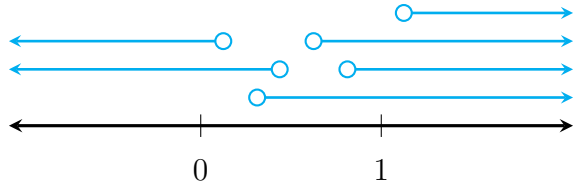

In general, a belief player $1$ can have about player $2$ where player $2$ can choose from $m$ possible actions $b_1,\ldots,b_m$ is an $m$-tuple $(q_1,\ldots,q_m) \in [0,1]^{m}$ where $\sum_{i=1}^{m} q_i = 1.$
\\
Letting $q_m = 1 - q_1 - \cdots - q_{m-1},$ the set of possible beliefs is represented by the $m-1$-dimensional simplex $S^{m-1}$ (including its interior) of points $x = (x_1,\ldots,x_{m-1}) \in [0,1]^{m-1}$ where $\sum_{i=1}^{m-1} x_i \leq 1.$
For example, when $m=2,$ this is just the line segment $[0,1].$ When $m=3,$ this is the closed triangle region bounded by the $x$-axis, the $y$-axis, and the line $x+y=1.$ When $m=4,$ it is a tetrahedron with vertices at $(0,0,0), (1,0,0), (0,1,0), (0,0,1),$ and so forth.
\\
\newline
Let us examine this for the $d=3$ case. We can plot the mixed strategies (beliefs) of player $2$ in $q_1$-$q_2$ space, which is $S^2,$ the right isosceles triangle in $\mathbb{R}^2$ bounded by $q_1, q_2 \geq 0, q_1+q_2 \leq 1.$  Then, the expected payoff for player $1$ by choosing action $a_i,$ with belief $q = (q_1,q_2,q_3=1-q_1-q_2),$ is a function $f_i : S^{1} \subset [0,1]^2 \to \mathbb{R}$ defined by $f(q_1,q_2) = a_{i,1} q_1 + a_{i,2} q_2 + a_{i,3} (1-q_1-q_2).$ Thus, the graph of $f_i$ is the restriction of a plane to $S^1.$
\\
\newline
Finally, we have the graph of $g,$ which is also a plane. Again, we can consider $f' = f-g,$ and the resulting functions are still planes, so we have $\max\{f'(x)\} > 0$ for all $x \in S^1.$
Now, each of the planes $f_i '$ will intersect the plane $\mathbb{R}^2 \times \{0\}$ at some line $l$ in the $q_1$-$q_2$ plane. First, if any $f_i'$ is parallel to $\mathbb{R}^2$, then clearly it is constant and strictly greater than zero, and so that single function $f_i > g$ for all $x \in S^2$.  This means that the pure strategy $a_i$ strictly dominates $g$ so we may ignore this case.
The cross sections will be lines that intersect the closed simplex $S^2$ (i.e. the closed triangle with vertices $(0,0), (0,1), (1,0)$) and we will have open half planes $h_1,h_2,\ldots,h_k$ that cover $S^2.$
\\
\newline
So in the general case, each action $a_i$ has an expected utility $f_i : S^{m-1} \to \mathbb{R}$ whose graph is a $m-1$ dimensional hyperplane in $\mathbb{R}^n,$ restricted to the domain $S^{m-1}.$ After taking a transformation $f_i' = f_i - g,$ the condition turns into $\max \{f_i' \} > 0.$ Then, each hyperplane $f_i'$ intersects $\mathbb{R}^{m-1} \times {0}$, which cuts $\mathbb{R}^{m-1}$ into two half spaces, where one of the (open) half spaces represents the set of points where $f_i' > 0$, and the other (closed) half space represents the set of points where $f_i' \leq 0.$ Of course, we only care about the intersection of the half spaces with the domain $S^{m-1} \subset \mathbb{R}^{m-1}.$
\\

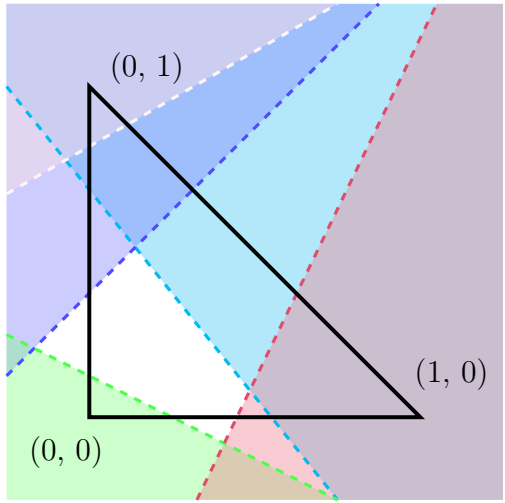
\begin{figure}[H]
    \centering
    \begin{tikzpicture}[scale=1.1]
            
            \fill[cyan!60, opacity=0.4] (-1, 5) -- (-1, 4) -- (3, -1) -- (5,-1) -- (5, 5) -- cycle;
            
            \fill[blue!50, opacity=0.4] (-1,5) -- (-1, 0.5) -- (3.5, 5) -- cycle;

            \fill[green!50, opacity=0.4] (-1, 1) -- (-1, -1) -- (3,-1) -- cycle;


            \fill[purple!50!red!50, opacity=0.4] (1.3, -1) -- (5, -1) -- (5, 5) -- (4.2, 5) -- cycle;


            \fill[pink!50!purple!20, opacity=0.5] (-1, 5) -- (-1, 2.7) -- (3,5) -- cycle;

            \draw[dashed, cyan!80, very thick] (-1, 4) -- (3, -1);

            \draw[dashed, blue!70, very thick] (-1, 0.5) -- (3.5, 5);

            \draw[dashed, green!70, very thick] (-1, 1) -- (3,-1);

            \draw[dashed, purple!70!red!70, very thick] (1.3, -1) -- (4.2, 5);

            \draw[dashed, pink!100!purple!15, very thick] (-1,2.7) -- (3,5);

            \draw[black, line width=1.5pt] (0, 0) -- (4, 0) -- (0, 4) -- cycle;

            \node[above right] at (-0.1, 3.9) { \hspace{0.1 cm} {(0, 1)}};

            \node[below left] at (0.3, -0.1) {{(0, 0)}};

            \node[above right] at (3.8, 0.2) {{(1, 0)}};
    
    \end{tikzpicture}
        
    \caption{A partial covering of $S^2$ with half planes}
\end{figure}

We are thus left with the following problem:

\begin{prop}
\textbf{(Theorem 2.1, geometric form)} \,
Let $Q \subset \mathbb{R}^d$ be a $d$-dimensional convex polytope (including its boundary and interior). Suppose a set of open half spaces $h_1,h_2,\ldots,h_n$ covers $Q.$ Then, there exists a subset of at most $d+1$ of these $h_i$ that also cover $Q.$
\end{prop}

The key observation is that we want to remove redundant half spaces. If a half space covers a region that is also covered by another set of half planes, then we can remove the first half plane and the covered region will be the same. A half space being covered by another half space corresponds to dominance by a pure strategy. Being covered by a set of half spaces corresponds to being dominated by a mixed strategy. We define a \textbf{minimal configuration} to be one where no half plane can be removed and the remaining half planes still cover the polytope, i.e. each half plane contains a point that is uniquely covered by it.

First, we gain some intuition for this result by showing it for $d=3,$ i.e. a covering of $S^2$ with half planes. \\

\begin{prop}
\textbf{(Theorem 2.1 for $d=3$)} \,
Let $Q \subset \mathbb{R}^2$ be a closed convex polygon including its interior. Suppose $n$ open half planes covers $Q.$ Then, there exists a subset of at most $3$ half planes that cover $Q.$
\end{prop}

\begin{proof}
Suppose for contradiction that there exists a covering of the polygon $Q$ with half planes $\{h_i\}$ such that any subset of $\{h_i\}$ that also coves $Q$ contains at least four planes. Consider any minimal configuration $\{h_j'\} \subset \{h_i\}$ that covers $Q.$ By assumption, $| \{h_j \} | \geq 4.$
\\
\newline
Since $\{h_j'\}$ is minimal, we cannot remove any $h \in \{h_j'\}$ to get a smaller subset of half planes that also covers $Q.$ Thus, for every $h_j',$ there exists some point $p_j \in Q$ that is only covered by $h_j'.$ \\
\newline
Thus, we have four points $p_1,p_2,p_3,p_4 \in Q$ such that $p_j$ is uniquely covered by $h_j'$ for all $1 \leq j \leq 4.$ 
\\
\newline
There are a few possible configurations. We check the cases based on the convex hull $H$ of $p_1,p_2,p_3,p_4.$
\\
\newline
Suppose $H$ is a line segment, i.e. $p_1 p_2 p_3 p_4$ are collinear, in that order. Then, we've just reduced it to the $d=2$ one dimensional case, as each plane bisects the line at some intersection point and covers everything to one side. So we see that there must be some plane that covers at least two points as there are more than two planes.
\\
\newline
Now suppose $H$ is a triangle, say $H = \triangle p_2 p_3 p_4.$ Then $p_1$ lies in $\triangle p_2 p_3 p_4.$
If $p_1$ lies strictly in the interior of $\triangle p_2 p_3 p_4,$ then clearly any half planes that cover $p_1$ also intersect $\triangle p_2 p_3 p_4,$ and we see that one of those vertices are necessarily covered as well. (If the plane doesn't cover any of the three vertices, then the entire triangle lies on the other side of the plane and is also not covered so $p_1$ is not covered, which can't happen.)
\\
\newline
If $p_2$ lies on the boundary (one of the sides) of $\triangle p_2 p_3 p_4,$ then we need to be a bit careful but the same result holds as in the previous subcase.
\\
\newline
Finally, suppose $H$ contains all four points $p_1,p_2,p_3,p_4,$ so the points form a non-degenerate convex quadrilateral, without loss of generality $p_1 p_2 p_3 p_4$ in clockwise order. Then, note that $p_1 p_3$ intersects $p_2 p_4$ at some point $x$ inside $H,$ and again we see that any point that covers $x$ also covers at least two of the vertices of $H,$ so this case is covered.
\\
\newline
The result follows, i.e. any covering of $S^2$ with $n$ half planes contains a subcover using at most three half planes.
\end{proof}

For the general case in $\mathbb{R}^d,$ we utilize \textbf{Radon's Theorem}, which states:

\begin{thm}
\textbf{(Radon)} \,
Let $P$ be a set of $n \geq d+2$ distinct points in $\mathbb{R}^d.$ Then, $P$ can be partitioned into two nonempty subsets $V_1$ and $V_2$ such that the convex hulls of $V_1$ and $V_2$ intersect.
\end{thm}

\begin{proof}
    See appendix.
\end{proof}

For the general case in $\mathbb{R}^d,$ let $Q \subset \mathbb{R}^d$ be a convex polytope (i.e. including the boundary and interior).
Suppose there are $n \geq d+2$ half spaces $h_1 \ldots h_n$ which cover $Q,$ such that each half space contains a point $p_i \in Q$ that is only covered by $h_i.$
\\
\newline
By Radon's Theorem (as $n \geq d+2$), we can partition $P = \{p_1,\ldots,p_n\}$ into two sets $P_1$ and $P_2$ such that their convex hulls $H(P_1), H(P_2)$ intersect. Consider a point $x \subset H_1(P_1) \cap H(P_2)$ in this intersection, which lies in the original polytope $Q$ by convexity: as $P_1,P_2 \subset P \subset Q,$ we have $H(P_1), H(P_2) \subset H(P) \subset H(Q) = Q.$
\\
\newline
However, $x$ cannot be covered by any $h_i$ with corresponding $p_i$ in $P_1,$ as $x$ is in $H_2,$ and so any such $h_i$ would have to intersect $H_2$ and contain some other point in $P_2.$ Analogously $x$ cannot be covered by any $h_i$ with $p_i$ in $P_2.$ So $x \in Q$ is not covered by $\bigcup h_i,$ which is a contradiction.
This uses the fact that for an open half space, if a set of points all lie on the other side, then their convex hull also lies entirely on the other side.  
\\
\newline
So letting $Q = S^{n-1},$ the $n-1$-dimensional simplex, we have proven \textbf{Theorem 2.1}. $\blacksquare$
\\

An immediate result of \textbf{Theorem 2.1} is that it gives us a characterization of which actions can be removed by iterated strict dominance, i.e. which ones are rationalizable. Indeed, as we've shown, if $f_i$ lies under the upper convex hull of $f_1,\ldots,f_m,$ then $f_i$ is strictly dominated. In practice, we should be able to solve this with linear programming. If $f_i$ is part of the upper convex hull at any point $q \in S^{m-1},$ then $f_i(q) \geq f_k(q)$ for all $1 \leq k \leq m,$ so no mixed strategy beats $f_i$ with belief $q.$ Finally, we point out the case of weak dominance. When $m=2,$ it's easy to see that either the graph of the function lies strictly below the upper convex hull, coincides with the upper convex hull on some set that contains an interval, or intersects the upper convex hull exactly at a single point. The last case is when the action is weakly dominated. For $m=3,$ we see that weak dominance can result in a weak set (where the graph of the function coincides with the upper convex hull) of a point or a line, and so forth.
\\

We also note that $\textbf{Theorem 2.1}$ still holds even if one player has infinitely many actions. 
 \begin{cor}\, Let $G$ be a two player game where player $2$ has a finite number of actions $B = \{b_1,\ldots,b_m\}$, and player $1$ has infinitely many actions $A = \{a_i\}_{i \in I}$ for some index set $I.$ Then, if some action $a_i$ is strictly dominated by a mixed strategy over a set of actions $\{ a_j \} = A' \subset A,$ in fact $a_i$ is strictly dominated by a mixed strategy over a finite set of at most $m$ actions. 
\end{cor}

\begin{proof}
We simply note that we are covering $Q = S^{n-1}$ with open half spaces $h_{a_j},$ and since $Q$ is compact (closed and bounded), $\bigcup_{a_j \in A'} h_{a_j}$ forms an open cover of $Q,$ so there exists a finite subcollection $a_1', a_2', \ldots, a_k' \in A'$ such that $\bigcup_{a_j' \in A'} h_{a_j'}$ also covers $Q.$ Then by the previous result, we again know that there exists a subset of at most $m$ of these $h_{a_j'}$ that covers $Q,$ and we are done.
\end{proof}

\bigskip

\section{A Linear Algebra Perspective}

Now, let us turn to a completely different representation of dominance. This approach yields a more natural, geometric, interpretation of \textbf{Theorem 2.1}, which can be seen fundamentally as a dimensionality result.
\\
\newline
Again suppose player $2$ has $m$ actions $b_1,\ldots,b_m.$
We view the payoffs for player $1$ as a matrix $A = (a_{ij} )$  (linear transformation), and each action $a_i$ as a row vector $v_i = (a_{i,1}, \ldots, a_{i,m}).$
\\
\newline
Assuming von-Neumann Morgenstern Preferences, by taking a positive affine transformation, we may first transform the matrix to $A' = A + cJ$, where $J$ is the $n \times m$ matrix with all $1$'s and $c$ is a positive constant, so that all coordinates are strictly positive.
\\
\newline
We view each payoff action as a vector in $\mathbb{R}^m.$ Then, if we play a mixed strategy over some actions $\{a_j\}$ with probabilities $\{p_j\},$ then in vector form this represents $\sum p_j v_j,$ where $\sum p_j = 1.$
\\
\newline
In other words, if we take the convex hull of all the $v_j,$ along with the origin, the possible mixed strategies over these vectors is represented by the outer boundary $B$ of the convex hull $H.$ 
\\
\newline
Similarly, we see that the convex polytope formed by taking the convex hull of the vectors is exactly the region
$$\left\{ \sum_{i=1}^{m} \lambda_i v_i \mid 0 \leq \lambda_1,\ldots,\lambda_m \leq 1 \text{ and } \sum_{i=1}^{m} \lambda_i \leq 1 \right\},$$ i.e. the region contained by the affine hull of the vectors $v_1,\ldots,v_m.$
\\
\newline
The natural connection here is that the characterization of probability distributions over $m$ elements using the simplex and the linear algebra formulation of convex hull of points using linear combinations are very similar, which leads to a re-interpretation of the probabilities/weights as coefficients of vectors.
\\
\newline
With this formulation in mind, we now prove several preliminary results about strict dominance.

\begin{prop}
A pure strategy $a_i$ is strictly dominated by another pure strategy $a_j$ if and only if the $k$th component of $v_i$ is strictly less than the $k$th component of $v_i$ for all $1 \leq k \leq m.$
\end{prop}

\begin{proof}
First, suppose that for some strategies $a_i, a_j,$ that $a_{i,k} < a_{j,k}$ for all components $1 \leq k \leq m.$
Clearly for any belief $q = (q_1,\ldots,q_m) \in [0,1]^m, \sum q_k = 1$ we have $\sum_{k=1}^{m} q_k (a_{i,k} - a_{j,k}) < 0$ because each term is less than or equal to zero and not all of them are zero, so $\sum_{k=1}^{m} q_k a_{i,k} < \sum_{k=1}^{m} q_k a_{j,k},$ i.e. $E(a_i) < E(a_j),$ and $a_i$ is strictly dominated by $a_j.$
\\
\newline
Now, suppose $v_i = (a_{i,1}, \ldots,a_{i,m})$ is strictly dominated by $v_j = (a_{j,1}, \ldots, a_{j,m}).$ Then, for all possible beliefs $q = (q_1,\ldots,q_m),$ we must have $\sum_{k=1}^{m} q_k a_{i,k} < \sum_{k=1}^{m} q_k a_{j,k},$ so \\ $\sum_{k=1}^{m} q_k (a_{i,k} - a_{j,k}) < 0.$ By letting $q = (1,0,\ldots,0), (0,1,0,\ldots,0),\ldots,(0,\ldots,0,1),$ we see that $a_{i,k} < a_{j,k}$ for all $1 \leq k \leq m.$ 
\end{proof}

\bigskip
It turns out that more generally, if we consider strict dominance by a mixed strategy, we can treat the mixed strategy as a pure strategy vector and compare them as if we were comparing two pure strategies' vectors. This leads us to the following proposition.

\begin{prop}
Consider a mixed strategy over actions $\{a_j\}$ with probability distribution $\{\lambda_j\},$ which we represent as $u = \sum \lambda_j v_j.$ Then, a pure strategy $a_i$ is strictly dominated by the mixed strategy over $\{a_j\}$ if and only if the $k$th component of $v_i$ is strictly less than the $k$th component of $u,$ i.e. $v_i$ is strictly dominated by the pure strategy vector representation of $\{a_j\}$ with lottery $\{\lambda_j\}.$
\end{prop}

\begin{proof}
Let $q \in [0,1]^m$ be any belief player $1$ can have about player $2$'s possible strategy. Then, the expected payoff with the pure strategy $a_i$ is simply the dot product $E(a_i) = \langle v_i, q \rangle,$ and the expected payoff by playing the mixed strategy is
$$E(\{a_j\}) = \sum \lambda_j \langle v_j, q \rangle = \sum \langle \lambda_j v_j, q \rangle = \left\langle \sum \lambda_j v_j, q \right\rangle = \langle u, q \rangle,$$
as with probability $\lambda_j,$ player $1$ plays action $a_j,$ in which case the expected payoff is $\langle v_j, q\rangle.$
Thus, we see that $E(a_i) < E(\{a_j\})$ if and only if $\langle v_i, q \rangle < \langle u, q \rangle$ for all possible beliefs $q,$ which from the previous claim is equivalent to the components of $v_i$ being strictly less than the components of $u.$
\end{proof}

\bigskip
Now, we can start classifying which types of pure strategies are dominated by a mix over other strategies. Let's look at a set of pure strategies $\{a_j\}$ with corresponding vectors $\{v_j\}.$ For example, consider the four vectors in $\mathbb{R}^2$ as below. Then, a good geometric intuition for whether a vector $v$ is strictly dominated by some mix over $\{v_j\}$  is if the vector is contained in the convex hull formed.
\\

For example, consider the following payoffs for player $1$ in a game where player $2$ has two actions and player $1$ has three actions.
A = $ \begin{pmatrix}
1 & 5 \\
5 & 1 \\
2 & 2 \\
\end{pmatrix}$
\\
\newline
We see that $a_3 = (2,2)$ is strictly dominated by a mix of over $a_1, a_2$ with equal probability for both actions, i.e. the mixed vector is $v' = \frac{1}{2} (1,5) + \frac{1}{2} (5,1) = (3,3)$ which strictly dominates $(2,2).$
\\
\newline
If we graph them, we see that $(2,2)$ lies between (inside) the vectors $(1,5)$ and $(5,1),$ i.e. the point $(2,2)$ lies inside the convex hull or triangle with vertices $(0,0), (1,5), (5,1).$ So a good intuitive guess is that points that lie within the convex hull formed by the vectors $(1,5), (5,1)$ will be dominated by some mix of the two.

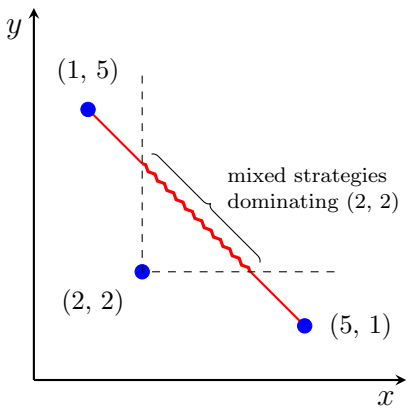
\begin{figure}[H]
\centering

\begin{tikzpicture}[>=stealth, every text node part/.style={align=center}, scale=0.9]
\draw[->, thick] (0,0) -- (5.5,0) node[below left] {$x$};
\draw[->, thick] (0,0) -- (0,5.5) node[below left] {$y$};

\draw[color=black, snake=brace, raise snake=5pt] (1.6,3.2) -- (3.2,1.6);

\node[above right] at (2.7,2.7) {\scriptsize{mixed strategies}};
\node[above right] at (2.7, 2.3) {\scriptsize{dominating (2, 2)}};


\draw[color=red, thick] (0.8, 4) -- (1.6, 3.2);
\draw[color=red, thick] (3.2, 1.6) -- (4, 0.8);

\draw[color=blue] (0.8,4) circle[radius=3pt];
\fill[color=blue] (0.8,4) circle[radius=3pt];

\draw[color=blue] (4, 0.8) circle[radius=3pt];
\fill[color=blue] (4, 0.8) circle[radius=3pt];

\draw[color=blue] (1.6, 1.6) circle[radius=3pt];
\fill[color=blue] (1.6, 1.6) circle[radius=3pt];

\draw[snake=zigzag, segment amplitude=0.6pt, segment length=5pt, color=red, very thick] (1.6, 3.2) -- (3.2, 1.6);

\draw[dashed, color=black] (1.6, 4.5) -- (1.6, 1.6 ) -- (4.5, 1.6);

\node[above] at (0.8, 4.2) {\footnotesize{(1, 5)}};
\node[right] at (4.2, 0.8) {\footnotesize{(5, 1)}};
\node[below left] at (1.5, 1.5) {\footnotesize{(2, 2)}};

\end{tikzpicture}
\caption{vector representation of dominance}
\end{figure}

And indeed, this easily follows from the previous claims.

\begin{prop}
If a strategy vector $v_i$ lies inside the convex hull (affine hull) of the vectors $H \left( \bigcup v_j \right)$ (but not on the outer faces/boundary), then the corresponding strategy $a_i$ is strictly dominated by a mix over $\{a_j\}.$
\end{prop}

\begin{proof}
By the linear combination definition of convex hull, this means that we can write $v_i = \sum \lambda_j v_j,$ where $0 \leq \lambda_j \leq 1$. Remove all vectors that have trivial contribution so we may assume that all $\lambda_j > 0.$ Note that $\sum \lambda_j = s < 1$ as the the vector $v_i$ does not lie on the outer boundary of the convex hull, so we may extend $v_i$ to a vector $u = \lambda v_i$ with $\lambda > 1$ that intersects the outer boundary of the convex hull.
\\
\newline
In particular, let $u = \sum \lambda_j' v_j = \sum \frac{\lambda_j}{s} v_j$ which represents a mixed strategy of $\{a_j\}$ with probability distribution $\{ \lambda_j / s \}$ as $\sum \frac{\lambda_j}{s} = 1.$ Note that  $\lambda_j< \frac{\lambda_j}{s} $ for all terms in the sum. Then have $$E(\{a_j\}) = \langle u, q \rangle = \left\langle \sum \frac{\lambda_j}{s} v_j, q \right\rangle > \langle \sum \lambda_j v_j, q \rangle = \langle v, q \rangle = E(a_i),$$
and $a_i$ is dominated by this mix over $\{a_j\}$ as claimed.
\end{proof}

\begin{prop} (Characterization of strategies that are strictly dominated by a mix over a given set) \\
Let $\{v_j\}$ be a set of pure strategy vectors corresponding to the actions $\{a_j\}.$ 
Consider the boundary $B,$ or the affine hull of the $\{v_j\}$ defined by 
$$B = \left\{ \sum \lambda_j v_j : 0 \leq \lambda_j \leq 1, \sum \lambda_j = 1 \right\},$$ which represents the set of all possible mixed strategies over $\{a_j\}.$
Then the set
$$X_B = \bigcup_{b \in B} \{ v : v_{(k)} < b_{(k)} \, \forall \, 1 \leq k \leq m \}$$ defines the set of pure strategies that are strictly dominated by some mix of $\{a_j\}.$ i.e. $v$ is strictly dominated by a mix over $\{v_j\}$ if and only if $v \in X_B.$
\end{prop}

\begin{figure}[H]
    \centering
    \includegraphics[scale=0.6]{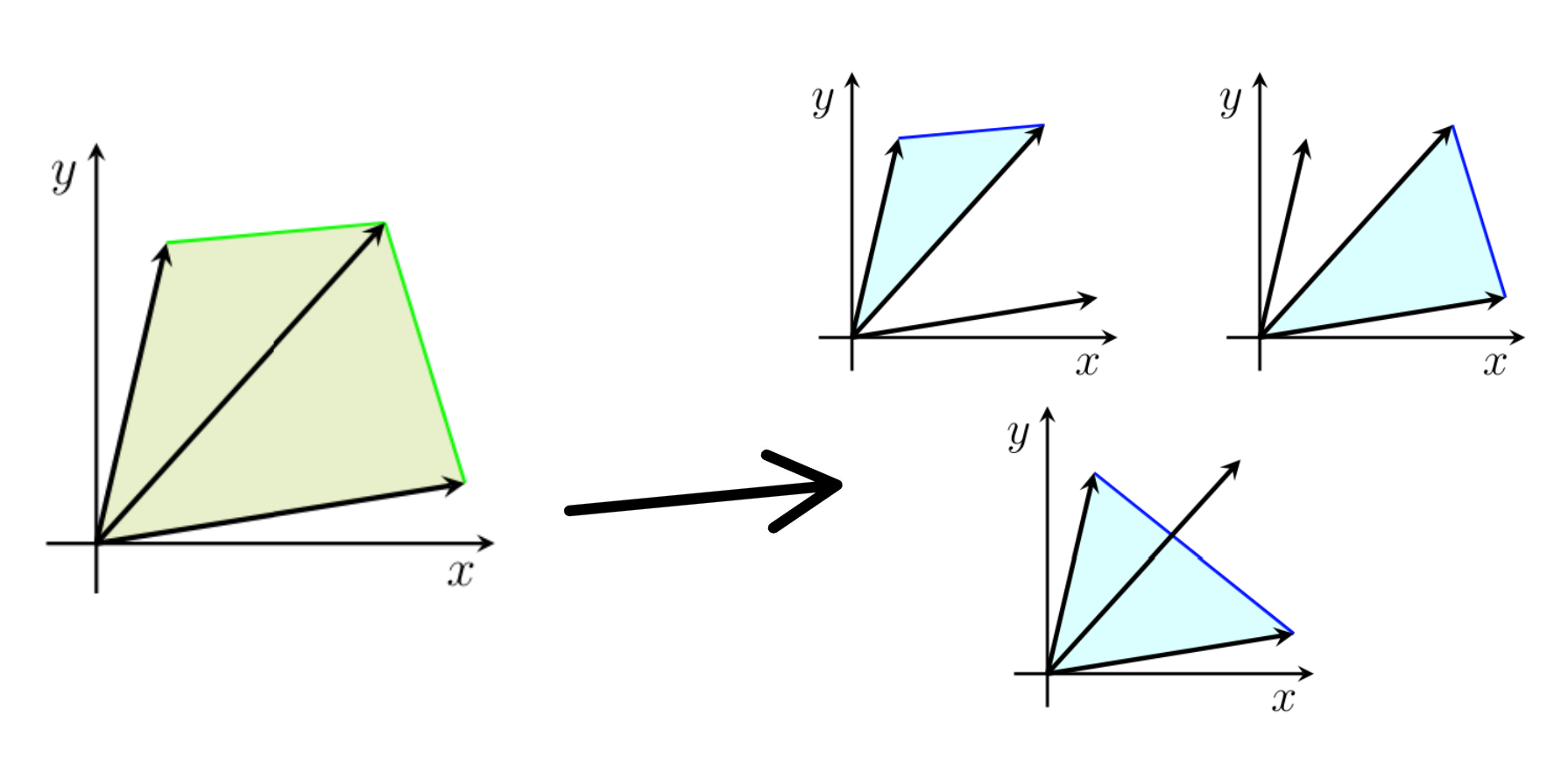}
    \caption{covering of green convex hull by blue triangles}
\end{figure}

Before using this characterization to provide another proof of \textbf{Theorem 2.1}, let us return to the case of points that lie strictly within the convex hull of some set of actions' vectors. We will use this to find an elegant and natural interpretation of \textbf{Theorem 2.1}. For now, we work with the closed convex hull (i.e. including the boundary $B$). The result is the same for open hull in the case of points that lie strictly within $H.$
\\

Let's look at four points $v_1,v_2,v_3,v_4.$ as shown below.
\\
\newline
Then, the region of vectors that are strictly dominated by a mix of $\{v_1,v_2,v_3,v_4\}$ is the region on the following page. \\

We can also look at the pairs of vectors, which bound triangular regions (above right). \\

These triangular regions represent the regions that are strictly dominated by two vectors. More specifically, the convex hulls are a subset of what points are strictly dominated by that set of vectors.
\\
\newline
Now, the motivation for \textbf{Theorem 2.1} becomes clear: the polygon region formed by the convex hull of the vectors $\{v_1,v_2,v_3,v_4\}$ is equal to the union of the convex hulls of pairs of the vectors.

\begin{figure}[H]
    \centering
    \includegraphics[scale=0.45]{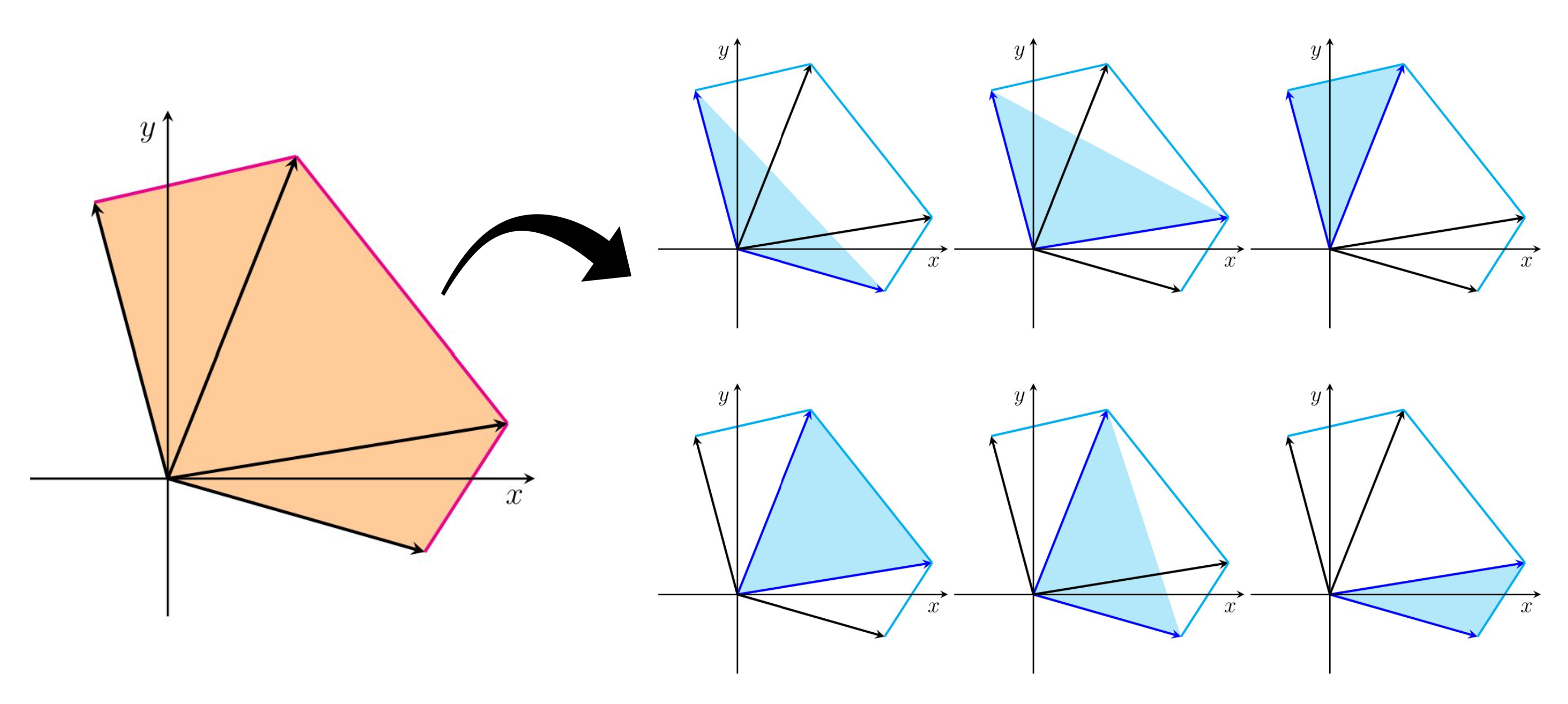}
    \caption{convex hull is covered by the union of triangles over all pairs of vectors}
\end{figure}
\bigskip

Formally, if $H(S)$ denotes the convex hull of a set of points $S,$ then our preceding observation is:
\begin{align*}
    H(O, v_1, v_2, v_3, v_4) &= H(O, v_1, v_2) \cup H(O, v_1, v_3) \cup H(O, v_1, v_4) \cup \ldots \cup H(O,v_3,v_4) \\
    &= \bigcup_{1 \leq i, j \leq 4} H(O,v_i, v_j)
\end{align*}

Similarly, if we have have $n$ points in $\mathbb{R}^3,$ i.e. $v_1,\ldots,v_n$ in $\mathbb{R}^3,$ it's easy to see that taking the union over all triples of points, i.e. $H(O,v_1,\ldots,v_n) = \bigcup_{i,j,k} H(O,v_i, v_j, v_k)$ yields the entire convex hull.
This is very intuitively clear. We are essentially considering all tetrahedrons with a fixed vertex $O$, and by shifting the vertices around, we use these tetrahedra to cover the entire interior of the polyhedra.
\\
\newline
More generally, a non-degenerate set of $d+1$ points in $\mathbb{R}^d$ should determine a $d$-dimensional figure, and by varying across different points in the set, we should be able to cover all $d$-dimensional points in the convex hull, which is the following statement:
Given $n$ nonzero points $v_1,\ldots,v_n$ in $\mathbb{R}^d,$ representing vectors extending from $O,$ then we have: \\ $H(O,v_1,\ldots,v_n) = \bigcup_{S_d \in \binom{[n]}{d}} H(O, v_{k_1}, \ldots, v_{k_d}),$ where $S_d$ varies over all subsets of \\ $[n] = \{1,\ldots,n\}$ with $d$ elements. \\

It is now easy to show that the generalization of our intuitive observation follows from $\textbf{Carathéodory's Theorem}.$

\begin{thm}
\textbf{(Carathéodory)}
If a point $x \in \mathbb{R}^d$ lies in the convex hull of a set $P,$ then $x$ can be written as the convex combination of at most d + 1 points in $P.$ i.e. there is a subset $P' \subset P$ consisting of $d + 1$ or fewer points such that $x \in H(P').$
\end{thm}

\begin{proof}
    See appendix.
\end{proof}
\bigskip

For our proof, we need a variant known as \textbf{Conical Carathéodory's Theorem} 
\begin{thm} \textbf{(Carathéodory's Theorem for Conical Hull)} \,
Let $n \geq d+1$ and $v_1,\ldots,v_n$ be $n$ distinct nonzero vectors in $\mathbb{R}^d.$ Consider the conical hull formed by this set.
Then, if a nonzero point $x$ is in the conical hull, $x$ can be written as the conical combination of at most $d$ of the $v_i.$

\end{thm}
\begin{proof}
    See appendix.
\end{proof} 
\bigskip

\begin{cor}
Furthermore, if $x$ also lies in the convex hull $H(O,v_1,\ldots,v_n),$ i.e. $x$ is a convex combination of the $v_i$ (and $O = 0,$ so sum of coefficients of the $v_i$ is less than or equal to one), then $x$ can be written as the convex combination of at most $d$ of the $v_i,$ and $O.$
\end{cor}
\begin{proof}
    See appendix.
\end{proof}


\bigskip

Finally, we can finish the second proof of $\textbf{Theorem 2.1}.$
Let $v_1,\ldots,v_n$ be the vector payoffs for player $1.$ Let $H = H(O,v_1,\ldots,v_n)$ be the convex hull of the set of vectors, and let $B$ be the outer boundary of $H.$
\\
\newline
From the characterization of these points, we know that a vector $v \in \mathbb{R}^n$ is strictly dominated by a mix over a subset of $\{v_j\} = v_1,\ldots,v_n$ if and only if there exists some $b \in B$ such that $v_{(k)} < b_{(k)}$ for all components $1 \leq k \leq m.$
But $B \subset H,$ and we know from the corollary following Carathéodory's Theorem that all points in $H$ can be represented as the affine sum of at most $d$ distinct vectors. Thus, if $v_i$ is strictly dominated by some such $b,$ we can write $b$ as the affine linear combination of at most $d$ of the $\{v_j\},$ which means that $v_i$ is strictly dominated by a mix over those $v_j,$ and in particular, is dominated by at most $d$ of the other vectors as claimed. $\blacksquare$ \\

We conclude by showing that the bound in \textbf{Theorem 2.1} is tight in the following sense:
\begin{prop} \,
    For every pair $(n,m) \in \mathbb{Z}^2_{\geq 2}$ there exists a game $G_{n,m}$ with $|A| = n$ and $|B| = m,$ and a strategy $a_i \in A$ that is strictly dominated by a mixed strategy over a set $\{a_j\} = A'$ of size $|A'| = \min( |A| - 1, |B| ),$ and $a_i$ is not strictly dominated by a mixed strategy over any smaller set of actions. An analogous statement holds for player 2.
\end{prop}

\begin{proof}
    We again represent pure strategies in their vector forms, with a pure strategy $a_i$ is associated to a vector $v_i \in \mathbb{R}^{m},$ and a mixed strategy over $\{a_j\}$ associated to a convex combination $\sum_{j} \lambda_j v_j.$  
    
    First suppose $|A|-1 < |B|.$ Then, consider
    \begin{align*}
        v_1 &= (n,0,0,\ldots,0,1,\ldots,1) \\
        v_2 &= (0,n,0,\ldots,0,1,\ldots,1) \\
        v_3 &= (0,0,n,\ldots,0,1,\ldots,1) \\
        & \ldots \\
        v_{n-1} &= (0,0,0,\ldots,n,1,\ldots,1) \\
        v_n &= (1,1,1,\ldots,1,0,\ldots,0) 
    \end{align*}
    where the $k$-th coordinate of $v_k$ is $n$ and the $j$-th coordinate is zero for $1 \leq j \leq n-1$ and is one for $j \geq n,$ for each $1 \le k \le n-1,$ and the first $n-1$ coordinates of $v_n$ are zero.
    Then, we see that $v_n$ is strictly dominated by the mixed strategy
    $$ \tfrac{1}{n-1} v_1 + \tfrac{1}{n-1} v_2 + \cdots + \tfrac{1}{n-1} v_{n-1} = (\tfrac{n}{n-1}, \tfrac{n}{n-1}, \ldots, \tfrac{n}{n-1},1,\ldots,1) \gg (1,1,1,\ldots,1,0,\ldots,0), $$
    but $v_n$ is clearly not strictly dominated by a mixed strategy over any smaller set as each $v_k$ for $1 \le k \leq n-1$ is needed in the mixed strategy.

    Now, suppose that $|B| \leq |A|-1.$ Then, consider 
    \begin{align*}
        v_1 &= (2m,0,0,\ldots,0) \\
        v_2 &= (0,2m,0,\ldots,0) \\
        v_3 &= (0,0,2m,\ldots,0) \\
         & \ldots \\
        v_{m} &= (0,0,0,\ldots,2m) \\
        v_{m+1} &= (1,1,1,\ldots,1) \\
        v_{k} &= (0,0,0,\ldots,0) \, \, \, \forall m+1 < k \leq n 
    \end{align*}

    We see that $v_{m+1}$ is strictly dominated by the mixed strategy $$\tfrac{1}{m} v_1 + \tfrac{1}{m} v_2 + \cdots + \tfrac{1}{m} v_m = (2,2,2,\ldots,2) \gg (1,1,1,\ldots,1),$$
    but $v_m$ is clearly not strictly dominated by any mixed strategy over a set of $m-1$ or fewer actions, as $v_k$ for $k \geq m+2$ are all zero vectors, and each $v_k$ for $k \leq m$ is needed as $v_k$ is the only vector that has a positive $k$-th coordinate.

\end{proof}

\section{Appendix}

The appendix provides proofs of Caratheodory's and Radon's Theorems, as well as a new proof of $\textbf{Theorem 1.1}$ using the representation of dominance in section 2 and half-spaces.

We begin the proof of $\textbf{Theorem 1.1}$ with the following covering lemma.

\begin{lemma}
    \textbf{(Rotation covering)}
    Suppose two half spaces $a \cdot x < 0$ ($A$) and $b \cdot x < 0$ ($B$) together cover a compact convex set $S$ in $\mathbb{R}^n$. Then there exists $\lambda \in [0,1]$ such that the half space $C$ given by $(\lambda a + (1 - \lambda) b) \cdot x < 0$ also covers $S.$
\end{lemma}

\begin{proof}
Suppose neither $A$ nor $B$ covers $S$ by itself so $\lambda \in (0, 1).$ Also assume that $a$ and $b$ are not parallel (linearly dependent) as that case can be easily handled.

Thus the hyperplanes $a \cdot x = 0$ and $b \cdot x = 0$ intersect at an $n-2$ dimensional subspace $P.$ Note that $P$ lies entirely outside of $S.$ 

These two hyperplanes divide $\mathbb{R}^n$ into four quadrant regions based on the signs of $a \cdot x$ and $b \cdot x.$
For the region where $a \cdot x < 0$ and $b \cdot x < 0$ clearly $(\lambda a + (1 - \lambda) b) \cdot x < 0.$ So consider the other two regions. The parts of $S$ in them are $S_A = (S \cap A) \setminus B$ and $S_B = (S \cap B) \setminus A$ which lie in two different quadrants. By assumption that neither $A$ nor $B$ alone cover $S,$ both $S_A, S_B$ are nonempty. Also $S_A, S_B$ are compact and convex. 

Consider the plane $C_{\lambda}$: $(\lambda a + (1 - \lambda) b) \cdot x = 0$ as $\lambda$ increases from zero to one. Note that $C_0$ intersects $S_A$ and $C_1$ intersects $S_B.$ Suppose for contradiction that $C_{\lambda}$ always intersects at least one of $S_A$ and $S_B.$ As $S_A, S_B$ are both closed, this means for some $\lambda^*$ the plane intersects both $S_A$ and $S_B,$ say at points $p$ and $q.$ 

We claim that the line segment $\overline{pq}$ intersects $P.$ Note that $C_{\lambda^*} \subset P \cup Q_1 \cup Q_2$ where $Q_1 = \{x : a \cdot x \geq 0, b \cdot x < 0\}$ and $Q_2 = \{x : b \cdot x \geq 0, a \cdot x < 0\}$ where we used $\lambda \in (0, 1).$ Since $p \in Q_1, q \in Q_2,$ as as we move from $p$ to $q$ along the segment, the signs of $a \cdot x$ and $b \cdot x$ each change (from nonnegative to negative, or vice versa). But this can only happen at $P$, which means that $\overline{pq}$ intersects $P$ at some point $z.$

However as $p, q \in S$, by convexity, this means $z \in S.$ This contradicts the fact that $P$ does not intersect $S.$
\end{proof}

\begin{thm}
\textbf{(Equivalence of rationalizability and dominance)}
    Given a finite, two player strategic game $G,$ the set of (pure) strategies that remain after iterated elimination of strictly dominated strategies are precisely the strategies that are rationalizable in $G.$
\end{thm}
\begin{proof}
    It is easy to show that any strictly dominated strategy is not rationalizable, so we focus on the reverse implication. As before, let player 1 have strategies $\{a_1, \ldots, a_m\}$ and we denote by $S^n = \{x \in \mathbb{R}^n \mid x_1 + \cdots + x_n = 1, x_i \geq 0 \, \, \forall i\}$ the $n-1$ dimensional simplex of probability vectors where $n$ is the number of actions of player 2. Thus, a belief player 1 has about player 2 is some $q \in S^n.$ \\
    
    As in the framework of section 2, we can represent each pure strategy $a_j$ as a linear function $f_j : S^n \to \mathbb{R},$ whose image is a restricted hyperplane in $\mathbb{R}^{n+1}$. Suppose the pure strategy $a_i$ represented by $f_i = g$ is not rationalizable. Again by taking $f_j' = f_j - g$ for all $j$ we may assume that $g$ is identically $0$ on $S^n$. For a slight abuse of notation, in the following we write $a_i$ to mean the payoff vector $(a_{i,1}, \ldots, a_{i,n})$, so the expected payoff of playing the $i$-th action under belief $q$ is $f_j(q) = a_i \cdot q$. \\

    For a given $f_j$ the set of beliefs $q$ for which $a_j$ yields strictly higher payoff than $a_i$ are those that satisfy $f_j(q) = a_j \cdot q > 0$, ie. the intersection of an open half-space in $\mathbb{R}^n$ with $S^n.$ Let this half-space be $H_j$.  \\

    If $a_i$ is a never best response, then the union of these $H_j$ over all $j \neq i$ must cover $S^n$ entirely. Similarly, for a mixed strategy $\sigma = r_1 a_1 + \cdots + r_m a_m$, the beliefs for which $\sigma$ yields higher payoff than $a_i$ are those $q$ satisfying $(r_1 a_1 + \cdots + r_m a_m) \cdot q > 0$, which is also a half-space $H_{\sigma}$. Then, $\sigma$ strictly dominates $a_i$ iff the half-space $H_{\sigma}$ covers $S^n$. \\
    
    The proof of the theorem follows from $\textbf{Lemma 4.2}$, which allows us to repeatedly replace a pair of (pure or mixed) strategies with a single mixed strategy, such that the union of the remaining half-spaces still covers $S^n$. This reduces the number of strategies while maintaining the covering invariant.
    Eventually, the process terminates with a single half-space $H^*$ that covers $S^n$ by itself, which corresponds to a mixed strategy $\sigma$ that dominates $a_i$.

    Indeed, suppose we have $m-1$ half spaces $p_j \cdot x < 0 \,$ $(H_j)$ for $j \in [m] \setminus \{i\}$ in $\mathbb{R}^n$ each representing a pure strategy $a_j$, and their union covers the compact convex set of probability vectors $S^n$.  We can now take two half spaces $p_1 \cdot x < 0$ and $p_2 \cdot x < 0$ and since $H_1 \cup H_2$ covers the convex compact set $S^n \setminus (H_3 \cup \cdots \cup H_m)$, there exists some $\lambda$ such that the half space $(\lambda p_1 + (1 - \lambda) p_2) \cdot x < 0$ covers $S^n \setminus (H_3 \cup \cdots \cup H_m),$ and we can replace $H_1, H_2$ with this halfspace instead. Thus, we have reduced the total number of half spaces while still covering $S^n.$
    Iterating this process, we eventually have a single half space of the form $\lambda_1 a_1 + \cdots + \lambda_m a_m < 0$ which covers $S^n.$ Since we take convex combinations in each step, the weights are such that $\lambda_1+\cdots+\lambda_m = 1$ with all $\lambda_i \geq 0$, so this defines a mixed strategy $\sigma$ which strictly dominates $a_i$.

\end{proof}

\begin{thm}
\textbf{(Radon)}
Let $P$ be a set of $n \geq d+2$ distinct points in $\mathbb{R}^d.$ Then, $P$ can be partitioned into two nonempty subsets $V_1$ and $V_2$ such that the convex hulls of $V_1$ and $V_2$ intersect.
\end{thm}

\begin{proof}
Consider any set $X = \{x_1, \ldots, x_{d+2}\} \subset \mathbb{R}^d$ of $d+2$ points.

Then there exists multipliers $a_1, \ldots, a_{d+2}$, not all zero, such that:
$$\sum_{i=1}^{d+2} a_i x_i = 0, \, \, \, \, \, \, \, \, \, \, \sum_{i=1}^{d+2} a_i = 0.$$

Take some particular nontrivial solution $a_1, \ldots, a_{d+2}$. Let $Y \subseteq X$ be the set of points with positive multipliers and $Z = X \setminus Y$ be the set of points with negative or zero multiplier. Observe that $Y$ and $Z$ forms a partition of $X$ into two subsets with intersecting convex hulls.

Indeed, $H(Y)$ and $H(Z)$ necessarily intersect because they each contain the point $$p = \sum_{x_i \in Y} \frac{a_i}{A} x_i = \sum_{x_i \in Z} \frac{- a_i}{A} x_i,$$ 

where $A = \sum_{x_i \in Y} a_i = - \sum_{x_i \in Z} a_i$. So we have expressed $p$ as convex combination of points in $Y$ and also as a convex combination of points in $Z$, which means that $p$ lies in the intersection of the convex hulls.

\end{proof}

\begin{thm}
\textbf{(Carathéodory)}
If a point $x \in \mathbb{R}^d$ lies in the convex hull of a set $P,$ then $x$ can be written as the convex combination of at most d + 1 points in $P.$ i.e. there is a subset $P' \subset P$ consisting of $d + 1$ or fewer points such that $x \in H(P').$
\end{thm}

\begin{proof}
Suppose $x$ can be written as a convex combination of $k > d+1$ of the points $x_i \in P.$ Without loss of generality, they are $x_1,\ldots, x_k.$ Thus, we can write $x = \sum_{j=1}^{k} \lambda_j x_j,$ where $0 < \lambda_j \leq 1$ and $\sum_{j=1}^{k} \lambda_j = 1.$
\\
\newline
Then, consider $x_2 - x_1, \ldots, x_k-x_1,$ which is a linearly dependent set, so there exists $\mu_j$ for $2 \leq j \leq k$ not all zero such that $\sum_{j=2}^{k} \mu_j(x_j-x_1) = 0$ and not all are zero. Then, if we set $\mu_1 = - \sum_{j=2}^{k} \mu_j,$ this means that $\sum_{j=1}^{k} \mu_j x_j = 0$ as well as $\sum_{j=1}^{k} \mu_j = 0,$ and not all of the $\mu_j$ are equal to zero, so there exists some $\mu_j > 0.$
\\
\newline
Now let $\alpha = \min_{1 \leq j \leq k} \left\{ \frac{\lambda_j}{\mu_j} : \mu_j > 0 \right\} = \frac{\lambda_i}{\mu_i}.$ Note that $\alpha > 0.$ Then, we have \\
$x = \sum_{j=1}^{k} \lambda_j x_j - \alpha \sum_{j=1}^{k} \mu_j x_j = \sum_{j=1}^{k} (\lambda_j - \alpha \mu_j) x_j$.
Since $\alpha > 0,$ then $\lambda_j - \alpha \mu_j \geq 0$ for all $j.$ \\
Also, $\sum_{j=1}^{k} (\lambda_j - \alpha \mu_j) = \sum_{j=1}^{k} \lambda_j - \alpha \sum_{j=1}^{k} \mu_j = 1 - \alpha \cdot 0 = 1.$ \\
However, by definition, $\lambda_i - \alpha \mu_i = 0.$
Thus, $x = \sum_{1 \leq j \leq k, j\neq i} (\lambda_j - \alpha \mu_j) x_j,$ where every $\lambda_j - \alpha \mu_j$ is non-negative and their sum is one.
\\
\newline
In other words, $x$ is a convex combination of $k-1$ points in $P.$ Repeating this process, we can write $x$ as a convex combination of at most $d+1$ points in $P,$ as desired. 
\end{proof}
\bigskip

\begin{thm} \textbf{(Carathéodory's Theorem for Conical Hull)} \,
Let $n \geq d+1$ and $v_1,\ldots,v_n$ be $n$ distinct nonzero vectors in $\mathbb{R}^d.$ Consider the conical hull formed by this set.
Then, if a nonzero point $x$ is in the conical hull, $x$ can be written as the conical combination of at most $d$ of the $v_i.$

\end{thm}
\begin{proof}
Suppose $x$ is expressed as the conical combination of $k > d$ of the $v_i,$ without loss of generality, $v_1,\ldots,v_k.$ Then, $x = \sum_{j=1}^{k} \lambda_j v_j,$ where $\lambda_j > 0$ for all $1\leq j \leq k.$ Since $k > d,$ then $v_1,\ldots,v_k$ is linearly dependent, so there exists $\mu_j,$ not all equal to zero, such that $\sum_{j=1}^{k} \mu_j v_j = 0.$ By multiplying each $\mu_j$ by $-1$ if necessary, we can ensure that at least one of these $\mu_j$ must be positive.
Without loss of generality, we may also assume that $\sum_{j=1}^{k} \mu_j \geq 0$ because if the sum were less than zero (then there would exist some negative $\mu_j),$ we could multiply each $\mu_j$ by $-1.$
Now, note that for any $\alpha,$ we have $\sum_{j=1}^{k} (\lambda_j - \alpha \mu_j) v_j = \sum_{j=1}^{k} \lambda_j v_j - \alpha \sum_{j=1}^{k} \mu_j v_j = x.$ We want to choose some $\alpha$ such that $\lambda_i - \alpha \mu_i = 0$ for some $i,$ and all $\lambda_j - \alpha \mu_j \geq 0.$

Indeed, if we let $\alpha = \min_{1 \leq j \leq k} \left\{  \frac{\lambda_j}{\mu_j} : \mu_j > 0    \right\} = \frac{\lambda_i}{\mu_i} > 0,$ then we do have $\lambda_j - \alpha \mu_j \geq 0$ for all $j.$ If $\mu_j \leq 0,$ then since $\lambda_j > 0,$ we have $\lambda_j - \alpha \lambda_j \geq 0.$ And if $\mu_j > 0,$ we have $0 < \alpha = \frac{\lambda_i}{\mu_i} \leq \frac{\lambda_j}{\mu_j},$ and again $\lambda_j - \alpha \lambda_j \geq 0.$
Thus, $x = \sum_{j=1}^{k} (\lambda_j - \alpha \mu_j) v_j$ represents a conical combination using at most $k-1$ of the $v_i,$ as all $\lambda_j - \alpha \mu_j \geq 0,$ and $\lambda_i - \alpha \mu_i = 0.$
We can repeat this process until we are left with at most $d$ of the $v_i,$ and so $x$ is the conical combination of at most $d$ of the $v_i$ as desired.
\end{proof} 
\bigskip
\begin{cor}
Furthermore, if $x$ also lies in the convex hull $H(O,v_1,\ldots,v_n),$ i.e. $x$ is a convex combination of the $v_i$ (and $O = 0,$ so sum of coefficients of the $v_i$ is less than or equal to one), then $x$ can be written as the convex combination of at most $d$ of the $v_i,$ and $O.$
\end{cor}
\begin{proof}
If $x$ lies in the convex hull $H(O = v_0,v_1,\ldots,v_k)$, that means that $x = \sum_{j=0}^{k} \lambda_j v_j,$ where $\lambda_j > 0$ and $\sum_{j=0}^{k} \lambda_j = 1.$ Since $O = 0,$ this is the same as saying that $\sum_{j=1}^{k} \lambda_j \leq 1.$
Again, if $k > d,$ then we can write $\sum_{j=1}^{k} \mu_j v_j = 0,$ where we can assume that $\sum_{j=1}^{k} \mu_j \geq 0.$
But then in the preceding proof, we have that $\sum_{j=1}^{k} (\lambda_j - \alpha \mu_j) = \sum_{j=1}^{k} \lambda_j - \alpha \sum_{j=1}^{k} \mu_j \leq \sum_{j=1}^{k} \lambda_j \leq 1,$ as $\alpha > 0$ and $\sum_{j=1}^{k} \mu_j \geq 0.$ Again, recall that $\lambda_i - \alpha \mu_i = 0,$ and each of the $\lambda_j - \alpha \mu_j \geq 0.$ \\

Thus, $x$ lies in the convex hull of $H(v_0,v_1,\ldots,v_{i-1},v_{i+1},\ldots,v_k).$
Like before, by repeating this process, we can show that $x$ lies in the convex hull $H^{*}$ of at most $d$ of the nonzero $v_i,$ along with $v_0,$ i.e. $x \in H(v_0, v_{j_1}, \ldots, v_{j_d}).$ 
\end{proof}

\newpage 

\textbf{Acknowledgements}

I would like to thank Scott Gehlbach at the University of Chicago Political Science department (and Harris' Political Economy program) for teaching me game theory during my first year SSI Formal Theory class. Prof. Gehlbach introduced the concepts of iterated dominance and rationalizability to me, and this paper grew out of some observations I made from examples in his class. I would also like to thank various UChicago economics and computer science professors for kindly taking a look at the paper and giving feedback during Winter/Spring 2022. Roger Myerson also provided some important suggestions about improving the details of the paper, and the last proposition showing the bound is tight is partly motivated by his observation.

\newpage

\textbf{References} \\

[1] Bernheim, D. (1984) Rationalizable Strategic Behavior. Econometrica 52: 1007-1028. \\

[2] Pearce, D. (1984) Rationalizable Strategic Behavior and the Problem of Perfection. Econometrica 52: 1029-1050. \\

[3] Fudenberg, D. and Tirole, J. (1993) Game Theory. Cambridge: MIT Press. \\

[4] Yildiz, M. Rationalizability lecture notes. https://dspace.mit.edu/handle/1721.1/99213 \\

[5] Obara, I. Rationalizability and Iterated Elimination of Dominated Strategies \\
http://www.econ.ucla.edu/iobara/Rationalizability201B.pdf

\bigskip
\bigskip
\bigskip
\bigskip
\bigskip
\bigskip
\bigskip
\bigskip
\bigskip
\bigskip
\bigskip
\bigskip
\bigskip
\bigskip
\bigskip
\bigskip
\bigskip



\end{document}